\documentclass[letterpaper, 10 pt, conference]{ieeeconf}  % Comment this line out if you need a4paper

\IEEEoverridecommandlockouts                              % This command is only needed if 
\title{\LARGE \bf
Interactive Function Computation with Reconstruction Constraints
}

\author{{\large Farideh Ebrahim Rezagah and Elza Erkip}\\
     \normalsize Dept. of ECE, Polytechnic Institute of NYU\\
     Email: fer216@nyu.edu, elza@poly.edu}

\usepackage{amsfonts}
\usepackage{amstext}
\usepackage{amssymb}
\usepackage{amsmath}
\usepackage{graphicx}

\newcommand{\Ta}{\mathbb{T}_A}
\newcommand{\Tb}{\mathbb{T}_B}
\newcommand{\Tj}{\mathbb{T}_\indx}

\newcommand{\X}{X}
\newcommand{\Y}{Y}
\newcommand{\alphX}{\mathcal{\X}}
\newcommand{\alphY}{\mathcal{\Y}}
\newcommand{\U}{U}
\newcommand{\alphU}{\mathcal{\U}}
\newcommand{\x}{x}
\newcommand{\y}{y}

\newcommand{\Extra}{S}
\newcommand{\alphExtra}{\mathcal{\Extra}}
\newcommand{\ExtraA}{\Extra_A}
\newcommand{\ExtraB}{\Extra_B}
\newcommand{\alphExtraA}{\alphExtra_A}
\newcommand{\alphExtraB}{\alphExtra_B}

\newcommand{\Z}{Z}
\newcommand{\W}{W}

\newcommand{\Za}{\Z_A}
\newcommand{\Zb}{\Z_B}
\newcommand{\Zj}{\Z_\indx}
\newcommand{\fn}{f}

\newcommand{\alphhatXa}{\hat{\mathcal{W}}_A}
\newcommand{\alphhatXb}{\hat{\mathcal{Z}}_B}
\newcommand{\alphhatYa}{\hat{\mathcal{Z}}_A}
\newcommand{\alphhatYb}{\hat{\mathcal{W}}_B}
\newcommand{\hatXa}{\hat{\W}_{A}}
\newcommand{\hatXb}{\hat{\Z}_{B}}
\newcommand{\hatYa}{\hat{\Z}_{A}}
\newcommand{\hatYb}{\hat{\W}_{B}}

\newcommand{\hatZj}{\hat{\Z}_{\indx}}

\newcommand{\hatWk}{\hat{\W}_{\indxx}}

\newcommand{\fj}{{\fn}_\indx(\X,\Y)}

\newcommand{\fA}{{\fn}_A(\X,\Y)}
\newcommand{\fB}{{\fn}_B(\X,\Y)}
\newcommand{\fjn}{{\fn}_\indx^{(n)}(\X^n,\Y^n)}
\newcommand{\fji}{{\fn}_\indx(\X_\blockindex,\Y_\blockindex)}

\newcommand{\dd}{d_\indx}
\newcommand{\dr}{d_{\indx\indxx}}

\newcommand{\da}{d_{A}}
\newcommand{\db}{d_{B}}
\newcommand{\dab}{d_{AB}}
\newcommand{\dba}{d_{BA}}
\newcommand{\Da}{D_{A}}
\newcommand{\Db}{D_{B}}
\newcommand{\Dab}{D_{AB}}
\newcommand{\Dba}{D_{BA}}

\newcommand{\round}{t}
\newcommand{\msg}{M}

\newcommand{\alphmsg}{\mathcal{M}}
\newcommand{\rate}{R}

\newcommand{\gd}{\Psi}
\newcommand{\gr}{\psi}
\newcommand{\gaX}{\gr_{A}}
\newcommand{\gbX}{\gd_{B}}
\newcommand{\gaY}{\gd_{A}}
\newcommand{\gbY}{\gr_{B}}

\newcommand{\gaXi}{\gr_{A,\blockindex}}
\newcommand{\gbXi}{\gd_{B,\blockindex}}
\newcommand{\gaYi}{\gd_{A,\blockindex}}
\newcommand{\gbYi}{\gr_{B,\blockindex}}

\newcommand{\blockindex}{i}

\newcommand{\msgindex}{j}

\newcommand{\indx}{k}
\newcommand{\indxx}{l}

\newcommand{\encod}{\phi}

\newcommand{\capreg}{\mathfrak{\rate}^{(\round)}}
\newcommand{\ratereg}{\mathcal{\rate}^{(\round)}}

\newcommand{\epn}{\epsilon_n}

\newcommand{\V}{V}
\newcommand{\covyj}{\kappa_{\y,\indx}}
\newcommand{\covxj}{\kappa_{\x,\indx}}
\newcommand{\D}{D}

\newtheorem{remark}{Remark}
\newtheorem{theorem}{Theorem}
\newtheorem{lemma}{Lemma}

\newtheorem{corollary}{Corollary}

\begin{document}

\maketitle
\thispagestyle{empty}
\pagestyle{empty}

%%%%%%%%%%%%%%%%%%%%%%%%%%%%%%%%%%%%%%%%%%%%%%%%%%%%%%%%%%%%%%%%%%%%%%%%%%%%%%%%
\begin{abstract}
  This paper investigates two-terminal interactive function computation with reconstruction constraints. Each terminal wants to compute a (possibly different) function of two correlated sources, but can only access one of the sources directly. In addition to distortion constraints at the terminals, each terminal is required to estimate the computed function value at the other terminal in a lossy fashion, leading to the constrained reconstruction constraint. A special case of constrained reconstruction is the common reconstruction constraint, in which both terminals agree on the functions computed with probability one. The terminals exchange information in multiple rate constrained communication rounds. A characterization of the multi-round rate-distortion region for the above problem with constrained reconstruction constraints is provided. To gain more insights and to highlight the value of interaction and order of communication, the rate-distortion region for computing various functions of jointly Gaussian sources according to common reconstruction constraints is studied.
\end{abstract}

%%%%%%%%%%%%%%%%%%%%%%%%%%%%%%%%%%%%%%%%%%%%%%%%%%%%%%%%%%%%%%%%%%%%%%%%%%%%%%%%
\section{Introduction}

The problem of computing functions of distributed correlated information sources arises in different networking scenarios, such as  sensor networks, cloud computing and smart grids. With limited communication resources, one important goal is to find the most effective way to operate the network in order to exchange minimal information while computing the desired functions with the required accuracy. For scenarios involving sensitive information, in addition to regular distortion constraints at the terminals, each terminal may wish to estimate the function computed at the other terminals within a certain accuracy. This leads to the \emph{constrained reconstruction constraint} \cite{Lapidoth2011}-\cite{Lapidoth2013}, which as a special case includes \emph{common reconstruction constraint} \cite{Steinberg2008} where all the terminals are required to agree on the computed function values with probability one.

The classical lossy source coding with side information problem investigated by Wyner and Ziv \cite{WynerZiv} involves one-way communication from the encoder to the decoder and only considers the distortion incurred in source reconstruction at the decoder. Steinberg \cite{Steinberg2008} considered an additional constraint in which the encoder is also required to reproduce the reconstruction at the decoder exactly, leading to the above common reconstruction constraint. The setting in \cite{Steinberg2008} suggests that while the side information at the decoder can be used to reduce the source rate through the \emph{binning} process as in the Wyner-Ziv setting, it cannot be further combined with the compression index, thereby resulting in a higher distortion at the decoder. Lapidoth et al \cite{Lapidoth2011} extended \cite{Steinberg2008} to allow for some distortion between the encoder's and the decoder's reconstruction, leading to the constrained reconstruction constraint.

While the above papers consider information flow from the encoder to the decoder only, Ma and Ishwar studied an interactive communication scenario where the two terminals exchange information in an alternating fashion for lossy function computation \cite{MaIshwarSomeResults}. The goal, as in the Wyner-Ziv setting, is to satisfy distortion requirements individually at each terminal. They showed that although for some functions there is no need to use multiple rounds of communication, there are functions for which every additional round of communication, subject to a fixed total sum rate, further decreases distortion.

In this paper we consider two terminal interactive function computation with reconstruction constraints, hence extending the framework in \cite{Lapidoth2011} to take into account interaction and in \cite{MaIshwarSomeResults} to take into account reconstruction constraints. We assume each terminal computes a different function of the two sources and reconstructs the function computed by the other terminal, all subject to distortion constraints. Multiple rounds of communication are allowed. Communication can take place in a sequential fashion, where terminals take turns, or in a simultaneous fashion.  We first identify the general rate-distortion region for this problem. We then study the Gaussian case with common reconstruction and compare sequential and simultaneous rate-distortion regions for different functions. Overall, our results highlight the importance of interaction under reconstruction constraints.

It is known that the coordination problem of \cite{Cuffdiss} is closely related to rate-distortion with side information. In particular, interactive coordination framework of \cite{Yassaee2012} can be used to study the problem investigated in this paper. Apart from giving an explicit characterization for the rate-distortion region with reconstruction constraints, we identify the individual set of rates for each communication round, while \cite{Yassaee2012} only considers sum rates.
%%%%%%%%%%%%%%%%%%%%%%%%%%%%%%%%%%%%%%%%%%%%%%%

\section{System Model}

We consider a two-terminal interactive function computation problem with $\round$ \textit{rounds} of communication as shown in Fig \ref{fig:setup}. Here the ``round'' refers to a one-way communication session between two terminals. Terminal A, $\Ta$, observes source $\X^n$, and terminal B, $\Tb$, observes $\Y^n$. The sources are drawn i.i.d $\sim p(\x,\y)$. The objective of $\Tj$ is to compute the function $\Zj^n=\fjn = (\fn_\indx(\X_1,\Y_1), \ldots, \fn_\indx(\X_n, \Y_n)), \indx=A,B$ in a lossy fashion and with reconstruction constraints \cite{Lapidoth2011}. Corresponding distortion measures and constraints will be discussed below.
 \begin{figure}[htbp]
   \centering
   \includegraphics[trim=8cm 9cm 8cm 0.1cm, clip, width=0.3\textwidth]{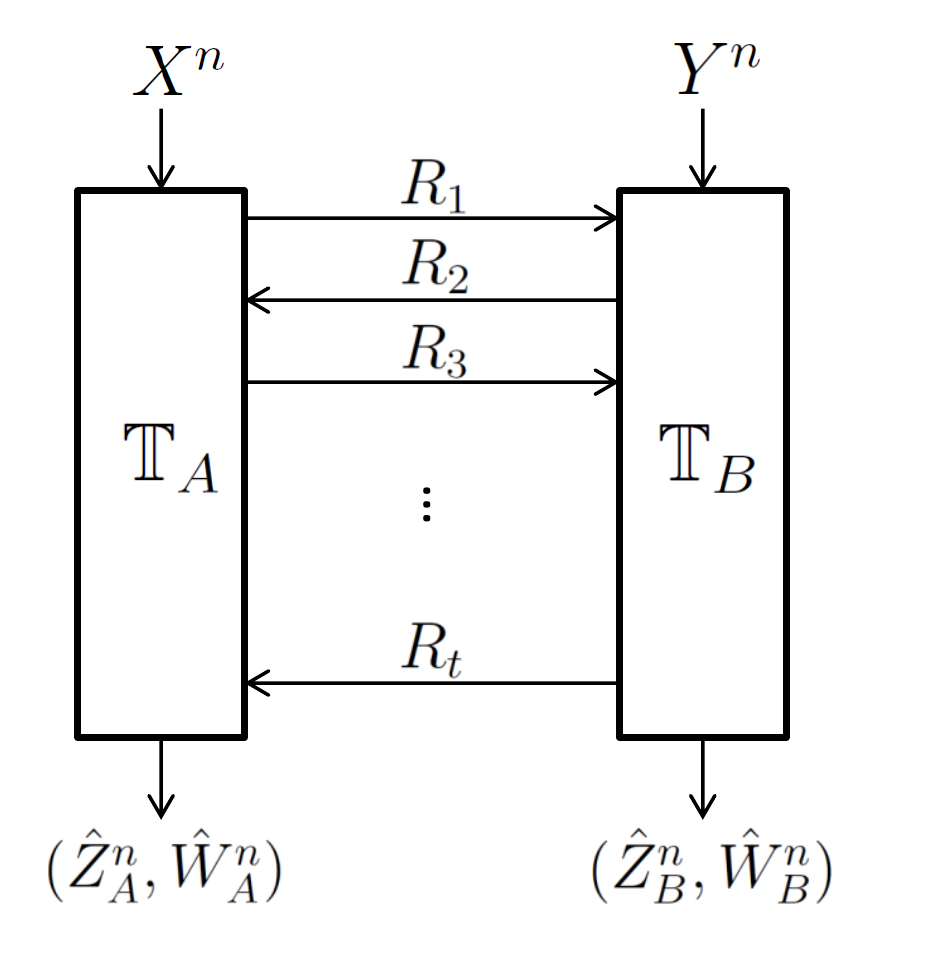}
   % where an .eps filename suffix will be assumed under latex,
   % and a .pdf suffix will be assumed for pdflatex
   \caption{Interactive function computation with $\round$ rounds of communication.}
   \label{fig:setup}
 \end{figure}

The messages sent in $\round$ rounds of communication are denoted as $\msg_1,\dots,\msg_{\round}$. These messages are sent sequentially starting from $\Ta$ (or $\Tb$). Alternatively, assuming $\round$ is even, the terminals can send their messages simultaneously, with the pair $(\msg_\msgindex, \msg_{\msgindex+1}), \msgindex=2\indx-1, \indx=1,2,\dots,\round/2$ transmitted at the same time. Below, we define encoding and decoding functions for sequential transmission starting from $\Ta$; definitions for sequential transmission starting from $\Tb$ or for simultaneous transmission similarly follow.

A message sent from $\Ta$ can depend on $\X^n$ and all previously communicated messages. In a similar manner each message sent from $\Tb$ can depend on $\Y^n$ and all previous messages. More formally, we can define encoding functions as follows:
\begin{IEEEeqnarray}{rCl}
  \msg_\msgindex=\encod^{(n)}_\msgindex&(\X^n,\msg^{\msgindex-1})&\hbox{  ,}\verb" " \msgindex\hbox{ odd;} \nonumber\\
  \msg_\msgindex=\encod^{(n)}_\msgindex&(\Y^n,\msg^{\msgindex-1})&\hbox{  ,}\verb" " \msgindex\hbox{ even;} \nonumber
\end{IEEEeqnarray}
where
\begin{IEEEeqnarray}{rCl}
  \encod^{(n)}_\msgindex&:\alphX^n \times \prod_{\indx=1}^{\msgindex-1}\alphmsg_\indx \mapsto\alphmsg_\msgindex&\hbox{  ,}\verb" " \msgindex\hbox{ odd;} \nonumber\\
  \encod^{(n)}_\msgindex&:\alphY^n \times \prod_{\indx=1}^{\msgindex-1}\alphmsg_\indx \mapsto\alphmsg_\msgindex&\hbox{  ,}\verb" " \msgindex\hbox{ even.} \nonumber
\end{IEEEeqnarray}
and $\alphmsg_\msgindex=\{1,2,\dots,2^{n\rate_\msgindex}\}$ where $\rate_\msgindex$ is the rate of the $\msgindex$'th message. We assume $\msg_0=1$ is deterministic. Note that $\encod_\msgindex$ corresponds to the encoding operations at $\Ta$, for $\msgindex$ odd, and encoding operations at $\Tb$, for $\msgindex$ even. After $\round$ rounds of communication, each terminal computes its desired function using $\gd_\indx, \indx=A,B$ as follows:
%\end{IEEEeqnarray}
\begin{IEEEeqnarray}{rCl}
  \hatYa^n=\gaY^{(n)}(\X^n,\msg^\round) \\
  \hatXb^n=\gbX^{(n)}(\Y^n,\msg^\round)
\end{IEEEeqnarray}
Furthermore, each terminal also computes an estimate of the other terminal's function using $\gr_\indx, \indx=A,B$, as follows
\begin{IEEEeqnarray}{rCl}
  \hatXa^n=\gaX^{(n)}(\X^n,\msg^\round) \\
  \hatYb^n=\gbY^{(n)}(\Y^n,\msg^\round)
\end{IEEEeqnarray}
The tuple $(\encod_1^{(n)}, \dots, \encod_\round^{(n)}, \gaY^{(n)}, \gbX^{(n)}, \gbY^{(n)}, \gaX^{(n)})$ is called an $(n,\rate_1,\dots,\rate_\round,\Da,\Db,\Dab,\Dba)$-\textit{code} if the produced sequences $\hatYa^n$, $\hatXb^n$, $\hatXa^n$ and $\hatYb^n$ satisfy
\begin{IEEEeqnarray}{l}
  \frac{1}{n}\sum_{\blockindex=1}^n E\da(\Za{}_{,\blockindex},\hatYa{}_{,\blockindex})\leq\Da \label{Adistortion}\\
  \frac{1}{n}\sum_{\blockindex=1}^n E\dab(\hatYa{}_{,\blockindex},\hatYb{}_{,\blockindex})\leq\Dab \label{Areconstruction}\\
  \frac{1}{n}\sum_{\blockindex=1}^n E\db(\Zb{}_{,\blockindex},\hatXb{}_{,\blockindex})\leq\Db \label{Bdistortion}\\
  \frac{1}{n}\sum_{\blockindex=1}^n E\dba(\hatXb{}_{,\blockindex},\hatXa{}_{,\blockindex})\leq\Dba \label{Breconstruction}
\end{IEEEeqnarray}
where $\dd(.,.)$ and $\dr(.,.)$, $\indx,\indxx =A,B$, $\indx\not=\indxx,$ are single letter distortion measures and $\Zj{}_{,\blockindex}=\fji$ is the actual value of the desired function at $\Tj$, $\indx=A,B$. Note that (\ref{Adistortion}) and (\ref{Bdistortion}) are the usual distortion constraints for computing $\fj$ at $\Tj$, $\indx=A,B$, whereas constraints (\ref{Areconstruction}) and (\ref{Breconstruction}) are the constrained reconstruction constraints \cite{Lapidoth2011}.
The rate and distortion tuple $(\rate_1,\dots,\rate_\round,\Da,\Db,\Dab,\Dba)$ is \textit{achievable} if for any $\epsilon\geq 0$ and $n$ sufficiently large there exist an $(n,\rate_1+\epsilon,\dots,\rate_\round+\epsilon,\Da+\epsilon,\Db+\epsilon,\Dab+\epsilon,\Dba+\epsilon)$-code. The set of all achievable $(\rate_1,\dots,\rate_\round,\Da,\Db,\Dab,\Dba)$ is denoted by $\capreg$ and the \textit{constrained-reconstruction rate-distortion region} for $\round$-round interactive communication is given by:
\begin{IEEEeqnarray}{c}
  \ratereg(\Da,\Db,\Dab,\Dba)\triangleq\nonumber\\
  \bigg\{(\rate_1,\dots,\rate_\round)|(\rate_1,\dots,\rate_\round,\Da,\Db,\Dab,\Dba)\in\capreg\bigg\}\nonumber
%  \bigg\{(\rate_1,\dots,\rate_\round)|(\rate_1,\dots,\rate_\round,\Da,\Db,\Dab,\Dba)\in\capreg\bigg\}\IEEEeqnarraynumspace
\end{IEEEeqnarray}

Similarly, we can define \textit{common-reconstruction rate-distortion region} $\ratereg_{CR}(\Da,\Db)$ by replacing conditions (\ref{Areconstruction}) and (\ref{Breconstruction}) with
\begin{IEEEeqnarray}{c}
  Pr(\hatYb^n\neq\hatYa^n)\leq\epn\label{A-CR}\\
  Pr(\hatXa^n\neq\hatXb^n)\leq\epn\label{B-CR}
\end{IEEEeqnarray}
where $\epn$ goes to zero as $n\rightarrow\infty$.
\begin{remark}
  The above set-ups differ from the interactive function computation problem of \cite{MaIshwarSomeResults} because of the additional reconstruction constraints (\ref{Areconstruction}) and (\ref{Breconstruction}), or (\ref{A-CR}) and (\ref{B-CR}). Also, they differ from \cite{Lapidoth2011} and \cite{Steinberg2008}, where only one round of communication was considered.
%  The above set-ups differ from the interactive function computation problem of \cite{MaIshwarSomeResults,MaIshwarInfMsg2009}-\cite{MaIshwarInfMsg2010} because of the additional reconstruction constraints (\ref{Areconstruction}) and (\ref{Breconstruction}) or (\ref{A-CR}) and (\ref{B-CR}). Also, they differ from \cite{Lapidoth2011} and \cite{Steinberg2008}, where only one round of communication was considered.
\end{remark}

\section{Main Result}

In this section, we provide a characterization of the rate-distortion region for $\round$-round interactive function computation with reconstruction constraints for general discrete memoryless sources. Hence we assume the alphabet sets $\alphX, \alphY, \alphhatYa, \alphhatXb, \alphhatXa$ and $\alphhatYb$ are all discrete.
\begin{theorem}\label{main1}
  For discrete memoryless sources $(\X,\Y)\sim p(\x,\y)$ and for sequential communication starting from $\Ta$, the constrained-reconstruction rate-distortion region $\ratereg(\Da,\Db,\Dab,\Dba)$ for computing functions $\fA$ and $\fB$ is given by the set of all rate vectors $(\rate_1,\dots,\rate_\round)$ such that for odd $\msgindex$
  \begin{IEEEeqnarray}{rcl}
    \rate_\msgindex\geq I(\X;\U_{\msgindex}|\Y,\U^{\msgindex-1}),\verb"  "\U_\msgindex-(\X,\ExtraA,\U^{\msgindex-1})-\Y \label{odd rate}
  \end{IEEEeqnarray}
  and for even $\msgindex$
  \begin{IEEEeqnarray}{rcl}
    \rate_\msgindex\geq I(\Y;\U_{\msgindex}|\X,\U^{\msgindex-1}),\verb"  "\U_\msgindex-(\Y,\ExtraB,\U^{\msgindex-1})-\X \label{even rate}
  \end{IEEEeqnarray}
  for some auxiliary random variables $\U_\msgindex\in\alphU_\msgindex$ $\msgindex=1,\dots,\round$, $\ExtraA\in\alphExtraA$ and $\ExtraB\in\alphExtraB$ where $\ExtraA-\X-\Y-\ExtraB$, such that there exist decoding functions:
  \begin{IEEEeqnarray}{rcl}
    \gaY:\alphX\times\alphExtraA\times\alphU^\round\mapsto\alphhatYa, \label{decoder1}\\
    \gbX:\alphY\times\alphExtraB\times\alphU^\round\mapsto\alphhatXb, \label{decoder2}\\
    \gaX:\alphX\times\alphExtraA\times\alphU^\round\mapsto\alphhatXa, \label{decoder3}\\
    \gbY:\alphY\times\alphExtraB\times\alphU^\round\mapsto\alphhatYb, \label{decoder4}
  \end{IEEEeqnarray}
  satisfying
  \begin{IEEEeqnarray}{l}
    E\da(\Za,\hatYa)\leq\Da, \label{Adistortion-thm}\\
    E\dab(\hatYa,\hatYb)\leq\Dab, \label{Areconstruction-thm}\\
    E\db(\Zb,\hatXb)\leq\Db, \label{Bdistortion-thm}\\
    E\dba(\hatXb,\hatXa)\leq\Dba, \label{Breconstruction-thm}
  \end{IEEEeqnarray}
  with $\Zj=\fj$, $\indx=A,B$, $\hatYa=\gaY(\X,\ExtraA,\U^\round)$, $\hatXb=\gbX(\Y,\ExtraB,\U^\round)$, $\hatXa=\gaX(\X,\ExtraA,\U^\round)$, and $\hatYb=\gbY(\Y,\ExtraB,\U^\round)$.
 Cardinality of the auxiliary random variable alphabets satisfy \begin{math}|\alphExtraA|\leq 2\end{math}, \begin{math}|\alphExtraB|\leq 2\end{math}, and \begin{math}|\alphU_\msgindex|\leq|\alphX||\alphExtraA|\prod_{\blockindex=1}^{\msgindex-1}|\alphU_\blockindex|+\round-\msgindex+5\end{math} for $\msgindex$ odd, and \begin{math}|\alphU_\msgindex|\leq|\alphY||\alphExtraB|\prod_{\blockindex=1}^{\msgindex-1}|\alphU_\blockindex|+\round-\msgindex+5\end{math} for $\msgindex$ even.

\end{theorem}
\begin{proof}
  The proof is provided in Appendix \ref{main1proof}.
\end{proof}

\begin{remark}
The region for sequential transmission starting from $\Tb$ or for simultaneous communication can be obtained through simple modifications of Theorem \ref{main1}.  Note that in general,  the constrained-reconstruction rate-distortion region depends on the order of communication. In particular for simultaneous communication, the rate constraints and the Markov conditions in (\ref{odd rate}) and (\ref{even rate}) would have to take into account that in each round the terminals transmit at the same time.
\end{remark}

\begin{remark}
  The auxiliary random variables $\ExtraA$ and $\ExtraB$ are defined by extending \cite{Lapidoth2013}, where a single round of communication with multiple distortion/reconstruction constraints is studied. In \cite{Lapidoth2013}, an auxiliary random variable is needed only in the reconstruction at the encoder. Due to the symmetry of our problem, and since both terminals wish to obtain reconstructions of functions estimated at the other terminal, we need to use these auxiliary random variables in all decoding functions.
\end{remark}

\begin{remark}
  Note that when the reconstruction constraints (\ref{Areconstruction-thm}) and (\ref{Breconstruction-thm}) are removed, cardinality arguments similar to \cite{Lapidoth2013} suggest that $|\alphExtraA|=|\alphExtraB|=1$. Therefore, the rate-distortion region in Theorem \ref{main1} coincides with that of \cite{MaIshwarSomeResults}.
\end{remark}

\begin{remark}
  Theorem \ref{main1} can easily be extended to include multiple functions of the sources. The cardinality of $\ExtraA$ and $\ExtraB$ are upper bounded by the total number of desired functions at both ends.
\end{remark}

\begin{corollary}\label{connection-corollary}
  For discrete memoryless sources $(\X,\Y)\sim p(\x,\y)$, and for sequential communication starting from $\Ta$, the common-reconstruction rate-distortion region $\ratereg_{CR}(\Da,\Db)$ for computing functions $\fA$ and $\fB$ is given by
  \begin{IEEEeqnarray*}{l}
     \ratereg_{CR}(\Da,\Db)=\ratereg(\Da,\Db,0,0).
  \end{IEEEeqnarray*}
  Furthermore, to compute $\ratereg_{CR}(\Da,\Db)$ directly for the nontrivial case where $\X\not=\Y$, we can use the region in Theorem \ref{main1} by removing constraints (\ref{Areconstruction-thm}) and (\ref{Breconstruction-thm}), while having the decoding functions depend only on $\U^\round$; i.e. $\hatYa=\hatYb=\gaY(\U^\round)$, $\hatXb=\hatXa=\gbX(\U^\round)$, with $|\alphExtraA|=|\alphExtraB|=1$.
\end{corollary}

\begin{proof}
  Gu \cite{Gu-diss}, and Jalali and Effros \cite{jalali2011} showed that in any network with any set of bounded distortion metrics, a rate is achievable for zero distortion between two of the variables, iff it is achievable with zero probability of error among the same variables. In other words, lossless reconstruction is achievable if and only if zero-distortion reconstruction is achievable. This includes our setting as a special case, where common reconstruction constraints become equivalent to constrained reconstruction with $\Dab=\Dba=0$.

  To compute $\ratereg_{CR}(\Da,\Db)$ directly, consider distortion metrics $\dr(.,.)$, $\indx,\indxx =A,B$, $\indx\not=\indxx,$  such that $\dr(\x,\y)=0$ iff $\x=\y$. Then (\ref{Areconstruction-thm}) and (\ref{Breconstruction-thm}) with $\Dab=\Dba=0$ imply that $\gaY(\X,\ExtraA,\U^\round)=\gbY(\Y,\ExtraB,\U^\round)$, with probability 1 and $\gbX(\Y,\ExtraB,\U^\round)=\gaX(\X,\ExtraA,\U^\round)$, with probability 1. Hence, we can argue that there exist functions $h_A(\U^\round)$ and $h_B(\U^\round)$ such that $\gaY(\X,\ExtraA,\U^\round)=\gbY(\Y,\ExtraB,\U^\round)=h_A(\U^\round)$ with probability 1 and $\gbX(\Y,\ExtraB,\U^\round)=\gaX(\X,\ExtraA,\U^\round)=h_B(\U^\round)$ with probability 1. Therefore, without loss of generality the decoders only depend on $\U^\round$. This also suggests that auxiliary random variables $\ExtraA$ and $\ExtraB$ are not necessary and we can set $|\alphExtraA|=|\alphExtraB|=1$.
\end{proof}

\section{Gaussian Sources}

  In this section, we illustrate the common reconstruction rate-distortion region, $\ratereg_{CR}(\Da,\Db)$, for jointly Gaussian sources. We set $d(a,b)=(a-b)^2$ for the distortion functions in (\ref{Adistortion}) and (\ref{Bdistortion}).
%We omit the details of extension of Theorem \ref{main1} and Corollary \ref{connection-corollary} to continuous alphabets.
We first consider linear functions of $\X$ and $\Y$. We then study extensions to other functions.
\begin{theorem}\label{gaussian1}
  Let $\X$ and $\Y$ be i.i.d. jointly Gaussian sources, with $\X\sim\mathcal{N}(0,\sigma_X^2)$ and $\Y=\X+\V$ where $\V\sim\mathcal{N}(0,\sigma_V^2)$ is independent of $\X$. For computing functions $\Za=\fA=\alpha_A\X+\beta_A\Y$ and $\Zb=\fB=\alpha_B\X+\beta_B\Y$, the $\round$-round common-reconstruction rate-distortion region starting from $\Ta$, $\ratereg_{CR}(\Da,\Db)$, is given by:
  \begin{IEEEeqnarray*}{rl}
  % \nonumber to remove numbering (before each equation)
    \sum_{m\text{ odd}}\rate_{m} \geq&\max_{\indx={A,B}}\bigg\{\frac{1}{2}\log\left(\frac{\alpha_\indx^2\sigma_\X^2\sigma_\V^2}{\D_\indx\sigma_\Y^2-\covyj^2}\right),\\
    &\frac{1}{2}\log \left(\left(\frac{\sigma_\X^2\sigma_\V^2}{\sigma_\Y^2}\right)\frac{\D_\indx+\alpha_\indx^2\sigma_\X^2-2\alpha_\indx\covxj}{\D_\indx\sigma_\X^2-\covxj^2}\right)\bigg\},\\
  \end{IEEEeqnarray*}
  \begin{IEEEeqnarray*}{rl}
    \sum_{m\text{ even}}\rate_{m} \geq&\max_{\indx={A,B}}\bigg\{\frac{1}{2}\log\left(\frac{\beta_\indx^2\sigma_\X^2\sigma_\V^2}{\D_\indx\sigma_\X^2-\covxj^2}\right),\\
    &\frac{1}{2}\log\left(\left(\sigma_\V^2\right)\frac{\D_\indx+\beta_\indx^2\sigma_\Y^2-2\beta_\indx\covyj}{\D_\indx\sigma_\Y^2-\covyj^2}\right)\bigg\},
  \end{IEEEeqnarray*}
  where $|\covxj|\leq\sqrt{\D_\indx\sigma_\X^2}$ and $|\covyj|\leq\sqrt{\D_\indx\sigma_\Y^2}$, $\indx=A,B$.
\end{theorem}
%%%%%%%%%%%%%%%%%%%%%%%%%%%%%%%%%%%%%%%%%%%%%%%%%%%
%%%%%%%%%%%%%%%%%%%%%%%%%%%%%%%%%%%%%%%%%%%%%%%%%%%
\begin{proof}
  We first find a lower bound on sum-rates in each direction using Corollary \ref{connection-corollary}. From $\Ta$ to $\Tb$ we have:
  \begin{IEEEeqnarray}{rCl}
    \sum_{m\text{ odd}}&\rate_{m}&\geq \sum_{\msgindex\text{ odd}}I(\X;\U_\msgindex|\Y,\U^{\msgindex-1})\nonumber\\
       &=&I(\X,\U^\round|\Y)=h(\X|\Y)-h(\X|\Y,\U^t)\nonumber\\
       &=&h(\X|\Y)+\max_{\indx=A,B}[\log|\alpha_\indx|-h(\alpha_\indx\X|\Y,\U^t)])\label{coefficient}\\
%       &=&h(\X|\Y)+\max_{\indx=A,B}[\log|\alpha_\indx|-h(\Zj-\hatZj|\Y,\U^\round)]\label{fn of U's 1}\label{achievableAB}\\
       &=&h(\X|\Y)+\max_{\indx=A,B}\bigg\{\log|\alpha_\indx|-h(\Zj-\hatZj|\Y,\U^\round),\nonumber\\
       &&\verb"     "\log|\alpha_\indx|-h(\Zj-\hatZj|\beta_\indx\Y-\hatZj,\U^\round)\bigg\} \label{fn of U's 2}\\
       &\geq&h(\X|\Y)+\max_{\indx=A,B}\bigg\{\log|\alpha_\indx|-h(\Zj-\hatZj|\Y),\nonumber\\
       &&\verb"     "\log|\alpha_\indx|-h(\Zj-\hatZj|\Zj-\hatZj-\alpha_\indx\X)\bigg\}\nonumber\\
       &\geq&\max_{\indx={A,B}}\bigg\{\frac{1}{2}\log \left(\frac{\sigma_\X^2\sigma_\V^2}{\sigma_\X^2+\sigma_\V^2}\frac{\D_\indx+\alpha_\indx^2\sigma_\X^2-2\alpha_\indx\covxj}{\D_\indx\sigma_\X^2-\covxj^2}\right),\nonumber\\
       &&\frac{1}{2}\log\left(\frac{\alpha_\indx^2\sigma_\X^2\sigma_\V^2}{\D_\indx(\sigma_\X^2+\sigma_\V^2)-\covyj^2}\right)\bigg\} \label{lapidoth}
  \end{IEEEeqnarray}
  where (\ref{coefficient}) holds because $h(\alpha\X)=\log|\alpha| h(\X)$, and (\ref{fn of U's 2}) holds because $\hatZj$ is a function of $\U^\round$. Finally, (\ref{lapidoth}) follows by upper bounding the conditional entropies using the fact that for any given covariance matrix, Gaussian random variables maximize the entropy. Similarly, we can find the sum-rate from $\Tb$ to $\Ta$.

  We can achieve any rate in this region by two rounds of communication using the $(\U_1,\U_2)$-pair $\U_1=\X+\widetilde{\U}_1$ and $\U_2=\Y+\widetilde{\U}_2$, where $\widetilde{\U}_1 - \X - \Y - \widetilde{\U}_2$ are all jointly Gaussian. Decoding is done via $\hatZj=\hatWk=\alpha_\indx\U_1+\beta_\indx\U_2$, $\indx,\indxx=A,B$, $\indx\not=\indxx$. Using these variables in Corollary \ref{connection-corollary} and letting $\covxj$ be $cov\big((\Zj-\hatZj),\X\big)$ and $\covyj$ be $cov\big((\Zj-\hatZj),\Y\big)$, we obtain the desired achievable region.
\end{proof}

\begin{remark} \label{2-rounds}
  The achievability proof of Theorem \ref{gaussian1} directly applies for simultaneous transmission as well. Hence for computing linear functions of Gaussian sources the order of communication is not important and one can achieve the optimal performance by two simultaneous rounds of communication. However, this is not true for arbitrary functions of Gaussian sources. Consider the functions $\fA=\fB\triangleq \mathbb{I}(\X\geq\Y)$ again for $(\X,\Y)$ jointly Gaussian and with common reconstruction constraint, where $\mathbb{I}(.)$ is the indicator function. With sequential communication starting from $\Ta$, by letting $\rate_1$ to grow unboundedly, the rate pair $(\rate_1,\rate_2)=(\infty,1)$ will achieve arbitrary small $\Da$ and $\Db$. However, the same rate pair will result in strictly positive distortions using simultaneous communication.
\end{remark}

 \begin{figure}[htbp]
   \centering
   \includegraphics[trim=5cm 9cm 5cm 9cm, width=0.25\textwidth]{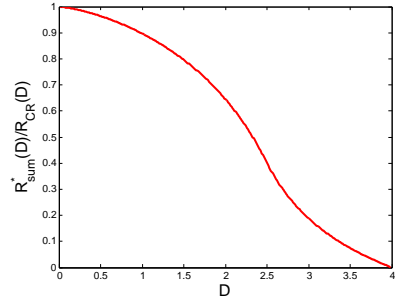}
   % where an .eps filename suffix will be assumed under latex,
   % and a .pdf suffix will be assumed for pdflatex
   \caption{Ratio of the interactive sum rate $\rate^*_{sum}(D)$, to one-way rate $\rate_{CR}(D)$ for computing $\fA=c$ and $\fB=\X$ for $\sigma_\X^2=4$ and $\sigma_\V^2=4$.}
   \label{fig:comparesteinberg}
 \end{figure}

  To illustrate the significance of interaction in achieving the desired common-reconstruction rate-distortion bounds, we compare it with the non-interactive setting in \cite{Steinberg2008}. Steinberg found the common reconstruction rate-distortion region for estimating $\X$ and when only one-way communication from $\Ta$ to $\Tb$ is allowed as:
\begin{IEEEeqnarray*}{l}
    \rate_{CR}(D)=\frac{1}{2}\log\left(\frac{\sigma_\X^2}{\sigma_X^2+\sigma_V^2}.\frac{D+\sigma_V^2}{D}\right).
\end{IEEEeqnarray*}
When the terminals are allowed to interact, the required sum-rate $\rate_{sum}(D)$ to achieve the same distortion can be found by setting $\fA=c$ and $\fB=\X$, with c being a constant, and by minimizing $\rate_{sum}(D)=\sum_{m\text{ odd}}\rate_{m}+\sum_{m\text{ even}}\rate_{m}$ of Theorem \ref{gaussian1}. Note that by the proof of Theorem \ref{gaussian1} and by Remark \ref{2-rounds}, having $\round=2$ rounds of interaction is sufficient. Let $\rate^*_{sum}(D)=\min \rate_{sum}(D)$. One can easily verify that:
\begin{IEEEeqnarray*}{l}
    \frac{\rate^*_{sum}(D)}{\rate_{CR}(D)}\leq 1
\end{IEEEeqnarray*}
%This inequality, however, does not grasp the significance of interaction.
As Figure \ref{fig:comparesteinberg} illustrates for $\sigma_\X^2=4$ and $\sigma_\V^2=4$, it can be shown that this ratio can get arbitrarily close to zero as $D$ increases, which suggests that interaction can significantly outperform one-way communication in the presence of common reconstruction constraints. This is because the common reconstruction constraint limits the use of the side information. In the achievable scheme of \cite{Steinberg2008}, side information at $\Tb$ enables binning of the compression indices, thereby reducing the required rate, but cannot be used in the decoding function to further improve the quality of the estimate. Therefore, as $\D$ increases, the dependence of the performance on side information decreases. However, if some small rate is allowed for $\Tb$ to communicate with $\Ta$ as in the interactive case, some limited amount of side information would be available at both $\Ta$ and $\Tb$, improving the performance significantly.

\section{Conclusion}

In this paper we have studied the interactive function computation with reconstruction constraints and characterized the constrained and common reconstruction rate-distortion regions. We have also evaluated the common-reconstruction rate-distortion region for linear functions of Gaussian sources. An interesting observation we have made is that two simultaneous rounds of communication are optimal in this case. A comparison with one-way communication has shown that a significant performance improvement is possible when interaction is allowed.

\appendices
\section{Proof of Theorem \ref{main1}}\label{main1proof}

The achievability proof is based on random coding and binning as in \cite{NIT} and is omitted.
We next provide a proof for the converse. The following lemma will be used without proof.
\begin{lemma}\label{convexity}
  \emph{[Convexity and monotonicity of rate-region]}
\begin{enumerate}
  \item $\ratereg(\Da,\Db,\Dab,\Dba)$ is convex in distortion vectors $(\Da,\Db,\Dab,\Dba)$.
  \item If $\D_\indx\leq\D'_\indx$ for $\indx\in\{A,B,AB,BA\}$ then
\end{enumerate}
\begin{IEEEeqnarray*}{l}
     \ratereg(\Da,\Db,\Dab,\Dba)\subseteq\ratereg(\Da',\Db',\Dab',\Dba').
  \end{IEEEeqnarray*}
\end{lemma}
The first part of the converse proof follows the standard steps in \cite{MaIshwarSomeResults}.
For an odd $\msgindex$ we have:
\begin{IEEEeqnarray}{rCl}
     n(\rate_\msgindex+\epsilon) &\geq& H(\msg_\msgindex) \nonumber\\
%     &\geq& H(\msg_\msgindex|\msg^{\msgindex-1},\Y^n) \nonumber\\
     &\geq& I(\X^n;\msg_\msgindex|\msg^{\msgindex-1},\Y^n) \nonumber\\
%     &=& H(\X^n|\msg^{\msgindex-1},\Y^n)-H(\X^n|\msg^{\msgindex},\Y^n) \nonumber\\
%     &=& \sum_{\blockindex=1}^{n}H(\X_\blockindex|\X^{\blockindex-1},\msg^{\msgindex-1},\Y^n) \nonumber\\
%     && -\sum_{\blockindex=1}^{n}H(\X_\blockindex|\X^{\blockindex-1},\msg^{\msgindex},\Y^n) \nonumber\\
     &=& \sum_{\blockindex=1}^{n}\bigg(H(\X_\blockindex|\X^{\blockindex-1},\msg^{\msgindex-1},\Y_\blockindex^n) \nonumber\\
     && -H(\X_\blockindex|\X^{\blockindex-1},\msg^{\msgindex},\Y_\blockindex^n)\bigg) \label{markov X_j}\\
%     &\geq& \sum_{\blockindex=1}^{n}H(\X_\blockindex|\X^{\blockindex-1},\msg^{\msgindex-1},\Y_\blockindex^n) \nonumber\\
%     && -\sum_{\blockindex=1}^{n}H(\X_\blockindex|\X^{\blockindex-1},\msg^{\msgindex},\Y_\blockindex^n) \label{removing condition 2}\\
%     &=&      \sum_{\blockindex=1}^{n}I(\X_\blockindex;\U_{\msgindex}|\U_{1,\blockindex},\U_{2},\dots,\U_{\msgindex-1},\Y_\blockindex) \label{Udef} \\
     &=&\sum_{\blockindex=1}^{n}I(\X_\blockindex;\U_{\msgindex}|\U_{1,\blockindex},\U_{2}^{\msgindex-1},\Y_\blockindex)  \label{Udef}
\end{IEEEeqnarray}
Where (\ref{markov X_j}) follows from chain rule and using the Markov chain $\X_\blockindex-(\X^{\blockindex-1},\msg^{\indx},\Y_\blockindex^n)-\Y^{\blockindex-1}$ for any $\indx=1,\dots,\round$ and any $\blockindex=1,\dots,n$ and (\ref{Udef}) holds by defining $\U_{1,\blockindex}\triangleq (\msg_1,\X^{\blockindex-1},\Y_{\blockindex+1}^{n})$ for $\blockindex=1,\dots,n$ and $\U_{\indx}\triangleq\msg_\indx$ for $\indx=2,\dots,\round$. Similarly for even $\msgindex$ we get:
\begin{IEEEeqnarray}{rCl}
% \nonumber to remove numbering (before each equation)
     n(\rate_\msgindex+\epsilon) &\geq& \sum_{\blockindex=1}^{n}I(\Y_\blockindex;\U_{\msgindex}|\U_{1,\blockindex},\U_{2}^{\msgindex-1},\X_\blockindex)
\end{IEEEeqnarray}
It is easy to see that $\ExtraA{}_{,\blockindex}=\X_{\blockindex+1}^n$ and $\ExtraB{}_{,\blockindex}=\Y^{\blockindex-1}$ satisfy Markov chains $\ExtraA{}_{,\blockindex}-\X_\blockindex-\Y_\blockindex-\ExtraB{}_{,\blockindex}$ and also%$\U_{1,\blockindex}-(\X_\blockindex,\ExtraA{}_{,\blockindex})-\Y_\blockindex$, and $\U_{\msgindex}-(\X_\blockindex,\ExtraA{}_{,\blockindex},\U_{1,\blockindex},\U_2^{\msgindex-1})-\Y_\blockindex$ for $\msgindex>1$ odd, and similarly $\U_{\msgindex}-(\Y_\blockindex,\ExtraB{}_{,\blockindex},\U_{1,\blockindex},\U_2^{\msgindex-1})-\X_\blockindex$ for $\msgindex>1$ even,
\begin{IEEEeqnarray*}{ll}
  \U_{1,\blockindex}-(\X_\blockindex,\ExtraA{}_{,\blockindex})-\Y_\blockindex,\\
  \U_{\msgindex}-(\X_\blockindex,\ExtraA{}_{,\blockindex},\U_{1,\blockindex},\U_2^{\msgindex-1})-\Y_\blockindex, & \verb" "\hbox{ $\msgindex>1$ odd;}\\
  \U_{\msgindex}-(\Y_\blockindex,\ExtraB{}_{,\blockindex},\U_{1,\blockindex},\U_2^{\msgindex-1})-\X_\blockindex, & \verb" "\hbox{ $\msgindex>1$ even.}
\end{IEEEeqnarray*}
for any $\blockindex=1,\dots,n$.

Let $\gaYi^{(n)}$ be the function that maps $(\X^n,\msg^\round)$ to the $\blockindex^{th}$ symbol of $\gaY^{(n)}(\X^n,\msg^\round)$, and $\Da{}_{,\blockindex}=E\da(\Za{}_{,\blockindex},\gaYi^{(n)})$ be the corresponding distortion. Therefore, $\frac{1}{n}\sum_{\blockindex=1}^n\Da{}_{,\blockindex}\leq\Da$.
%\begin{IEEEeqnarray*}{l}
%\frac{1}{n}\sum_{\blockindex=1}^n\Da{}_{,\blockindex}\leq\Da.
%\end{IEEEeqnarray*}
Similarly let $\gbXi^{(n)}$, $\gaXi^{(n)}$, and $\gbYi^{(n)}$ be the $\blockindex^{th}$ symbol of the corresponding functions with single letter distortions $\Db{}_{,\blockindex}$, $\Dab{}_{,\blockindex}$, and $\Dba{}_{,\blockindex}$.
We define decoding functions
\begin{IEEEeqnarray*}{l}
  \gaYi(\X_\blockindex,\ExtraA{}_{,\blockindex},\U_{1,\blockindex},\U_{2}^\round)\triangleq\gaYi^{(n)}(\X^n,\msg^\round),\label{single decoder 1}
\end{IEEEeqnarray*}
\begin{IEEEeqnarray*}{l}
  \gbXi(\Y_\blockindex,\ExtraB{}_{,\blockindex},\U_{1,\blockindex},\U_{2}^\round)\triangleq\gbXi^{(n)}(\Y^n,\msg^\round),\label{single decoder 3}\\
  \gaXi(\X_\blockindex,\ExtraA{}_{,\blockindex},\U_{1,\blockindex},\U_{2}^\round)\triangleq\gaXi^{(n)}(\X^n,\msg^\round),\label{single decoder 4}\\
  \gbYi(\Y_\blockindex,\ExtraB{}_{,\blockindex},\U_{1,\blockindex},\U_{2}^\round)\triangleq\gbYi^{(n)}(\Y^n,\msg^\round).\label{single decoder 2}
\end{IEEEeqnarray*}
To finish the proof we define
\begin{IEEEeqnarray*}{l}
    \IEEEeqnarraymulticol{1}{l}{\rate_{\msgindex,\blockindex}(\Da{}_{,\blockindex},\Db{}_{,\blockindex},\Dab{}_{,\blockindex},\Dba{}_{,\blockindex})\triangleq}\\ \qquad
    \left\{
            \begin{array}{ll}
              I(\X_\blockindex;\U_{\msgindex}|\U_{1,\blockindex},\U_{2}^{\msgindex-1},\Y_\blockindex),&\hbox{$\msgindex$ odd;} \\
              I(\Y_\blockindex;\U_{\msgindex}|\U_{1,\blockindex},\U_{2}^{\msgindex-1},\X_\blockindex),&\hbox{$\msgindex$ even.}
            \end{array}
          \right.
\end{IEEEeqnarray*}
We have
\begin{IEEEeqnarray*}{l}
  \rate_{\msgindex,\blockindex}(\Da{}_{,\blockindex},\Db{}_{,\blockindex},\Dab{}_{,\blockindex},\Dba{}_{,\blockindex})\in\\
  \verb"     "\ratereg_\msgindex(\Da{}_{,\blockindex},\Db{}_{,\blockindex},\Dab{}_{,\blockindex},\Dba{}_{,\blockindex})
\end{IEEEeqnarray*}
where $\ratereg_\msgindex(\Da{}_{,\blockindex},\Db{}_{,\blockindex},\Dab{}_{,\blockindex},\Dba{}_{,\blockindex})$ is the $\msgindex^{th}$ component of $\ratereg(\Da{}_{,\blockindex},\Db{}_{,\blockindex},\Dab{}_{,\blockindex},\Dba{}_{,\blockindex})$. Then
\begin{IEEEeqnarray}{rCll}
% \nonumber to remove numbering (before each equation)
  n(\rate_\msgindex+\epsilon)&\geq& \sum_{\blockindex=1}^{n}&\rate_{\msgindex,\blockindex}(\Da{}_{,\blockindex},\Db{}_{,\blockindex},\Dab{}_{,\blockindex},\Dba{}_{,\blockindex}) \nonumber\\
     &\geq& n \rate_\msgindex\bigg(&\frac{1}{n}\sum_{\blockindex=1}^{n}\Da{}_{,\blockindex},\frac{1}{n}\sum_{\blockindex=1}^{n}\Db{}_{,\blockindex},\frac{1}{n}\sum_{\blockindex=1}^{n}\Dab{}_{,\blockindex},\nonumber\\
     &&&\frac{1}{n}\sum_{\blockindex=1}^{n}\Dba{}_{,\blockindex}\bigg) \IEEEeqnarraynumspace\label{convexity of R}\\
     &\geq& n \rate_\msgindex(&\Da,\Db,\Dab,\Dba) \label{monotonicity of R}
\end{IEEEeqnarray}
Where (\ref{convexity of R}) and (\ref{monotonicity of R}) follow from Lemma \ref{convexity}, convexity and monotonicity of the rate-region, respectively. Therefore $\rate_\msgindex$ should also be in $\ratereg_\msgindex(\Da,\Db,\Dab,\Dba)$, completing the converse.

The cardinality of $\ExtraA$ and $\ExtraB$ can be upper bounded using techniques similar to \cite{Lapidoth2013}. Moreover, by considering super-sources $(\X,\ExtraA)$ and $(\Y,\ExtraB)$, and using similar arguments as \cite{MaIshwarSomeResults}, and taking into account that we have two extra reconstruction conditions to satisfy, we can get the cardinality bounds for $\U_\msgindex$, $\msgindex=1,\dots,\round$.

\bibliographystyle{IEEEtran}
\bibliography{ITWref}
% %\bibliography{definitions,bibliofile}
% %%
% %% where we here have assume the existence of the files
% %% definitions.bib and bibliofile.bib.
% %% BibTeX documentation can be obtained at:
% %% http://www.ctan.org/tex-archive/biblio/bibtex/contrib/doc/
% %%
% %%
% %%
% %% Or manual references (pay attention to consistency!):
% %\begin{thebibliography}{1}
% %\bibitem{shannon1948}
% %  C.~E. Shannon, ``A mathematical theory of communication,''
% %  \emph{Bell System Techn. J.}, vol.~27, pp. 379--423 and 623--656,
% %  Jul. and Oct. 1948.
% %\end{thebibliography}
%

\end{document}